\documentclass[12pt,english]{article}
\usepackage[english]{babel}
\usepackage[dvips]{graphicx}

\makeatletter
\usepackage{amssymb}
\usepackage{latexsym}
\usepackage{amsmath}
\usepackage{bbm}
\usepackage{amsthm}
\usepackage{paralist}

\def\vmin{v_{\min}}
\def\vmax{v_{\max}}
\def\bv{{\bf v}}

\usepackage{psfrag}

\newtheorem{theorem}{Theorem}[section]

\renewcommand{\Omega}{\mathrm{D}}

\makeatletter
\begin{document}

\title{On the Throughput-Delay Trade-off in Georouting Networks}

\author{
Philippe Jacquet\footnote{INRIA, Rocquencourt, France. Email: \texttt{name.surname@inria.fr}},
Salman Malik*,
Bernard Mans\footnote{Macquarie University, Sidney, Australia. Email: \texttt{bernard.mans@mq.edu.au}},
Alonso Silva*,
}
\date{}

\maketitle

\begin{abstract}

We study the scaling properties of a georouting scheme in a wireless multi-hop network of $n$ mobile nodes. Our aim is to increase the network capacity quasi linearly with $n$ while keeping the average delay bounded. In our model, mobile nodes move according to an i.i.d. random walk with velocity $v$ and transmit packets to randomly chosen destinations. The average packet delivery delay of our scheme is of order $1/v$ and it achieves the network capacity of order $\frac{n}{\log n\log\log n}$. This shows a practical throughput-delay trade-off, in particular when compared with the seminal result of Gupta and Kumar which shows  network capacity of order $\sqrt{n/\log n}$ and negligible delay and the groundbreaking result of Grossglausser and Tse which achieves network capacity of order $n$ but with an average delay of order $\sqrt{n}/v$. We confirm the generality of our analytical results using simulations under various interference models.

\end{abstract}

\section{Introduction}

Gupta and Kumar~\cite{gk2000} studied the capacity of wireless networks consisting of randomly located nodes which are immobile. They showed that if each source node has a randomly chosen destination node, the useful network capacity is of order $C\sqrt{n/\log n}$ where $n$ is the number of nodes and $C$ is the nominal capacity of each node. However, if the nodes are mobile and follow i.i.d. ergodic motions in a square area, Grossglauser and Tse~\cite{gt2002} showed that the network capacity can rise to $O(n C)$\footnote{We recall the following notation: (i) $f(n) = O(g(n))$ means that there exists a constant $c$ and an integer $N$ such that $f (n) \le cg(n)$ for $n >N$. (ii) $f (n) = \Theta(g(n))$ means that there exists two constants $c_1$ and $c_2$ and an integer $N$ such that $c_1g(n)\le f (n) \le c_2g(n)$ for $n >N$.} by using the mobility of the nodes. Note that in this case, a source node relays its packet to a random mobile relay node which transmits this packet to its destination node only when they come close together, {\it i.e.}, at a distance of order $1/\sqrt{n}$. Therefore, the time it takes to deliver a packet to its destination would be of order $\sqrt{n} L/v$ where $v$ is the average speed of the nodes and $L$ is the length of the fixed square area where nodes are deployed. In contrast, in Gupta and Kumar's result~\cite{gk2000}, the packet delivery delay tends to be negligible, although the network capacity drops by a factor of $\sqrt{n\log n}$. 

In this article, we aim to maximize the capacity of mobile networks while keeping the mean packet delivery delay bounded with increasing number of nodes. For relaying packets towards their destinations, mobile nodes use our proposed georouting strategy, called the {\em Constrained Relative Bearing} (CRB) scheme. We show that, in a random walk mobility model, this strategy achieves a network capacity of order $\frac{n}{\log n\log\log n}C$ with a time to delivery of order $L/v$. Our main contribution is summarized in Table~\ref{fig:comparison}. Note that in random walk mobility models, nodes have free space motion and move in straight lines with constant speed. This mobility model is a subclass of the free space motion mobility model. Therefore, we can also extend our result to mobility models where the average free space distance $\ell$ is non zero.  

{\small 
\begin{table}
\begin{center}
\begin{tabular}{|c|c|c|}
\hline
     & {\em Network Capacity}      & {\em Delivery Delay}\\
   \hline

    & & \\
Gupta \& Kumar       & $O\left(\sqrt{\frac{n}{\log n}}\right)$   &    negligible  \\
~\cite{gk2000}    & & \\
    \hline
    & & \\
Grossglauser \& Tse      &  $O(n)$  &   $O\left(\frac{\sqrt{n}}{v}\right)$ \\
~\cite{gt2002}    & & \\
       \hline
    & & \\
{\sl Our work}        & $O\left(\frac{n}{\log n\log\log n}\right)$  &   $O\left(\frac{1}{v}\right)$ \\
    & & \\
       \hline
\end{tabular}
\end{center}
\caption{\footnotesize Network Capacity vs. Delivery Delay Trade-off.}
\label{fig:comparison}
\end{table}
}

Consider an example of an urban area network in a fixed square area of length $L$ with number of nodes \mbox{$n=10^6$}, nominal bandwidth \mbox{$C=100$} kbps and delay per store-and-forward operation of $1$ ms. The average packet delivery delay for Gupta and Kumar's case would be around one second but with a network capacity of $10$ Mbps. In the case of Grossglauser and Tse, the network capacity would increase to about $100$ Gbps but if the straight line crossing time $L/v$ is about one hour ({\it e.g.}, with cars as mobile nodes), the time to delivery would be around one month. However, our model of using mobility of nodes along with the proposed CRB scheme, to relay packets to their destinations, would lead to a network capacity of $10$ Gbps with time to delivery of about one hour. 

This article is organized as follows. We first summarize some important related works and results in Section \ref{sec:rworks}. We discuss the models of our network and CRB scheme in Sections~\ref{sec:model} and \ref{sec:protocol} respectively. The analysis of capacity and delay can be found in Section \ref{sec:analysis} and we confirm this analysis using simulations in Section \ref{sec:simulations}. We also discuss a few  extensions of our work in Section \ref{sec:extension} and concluding remarks can be found in Section \ref{sec:conclusions}.

\section{Related Works}
\label{sec:rworks}

The main difference between the proposed models in the works of Gupta and Kumar~\cite{gk2000} and Grossglauser and Tse~\cite{gt2002} is that in the former case, nodes are static and packets are transmitted between nodes like ``hot potatoes'', while in the latter case, nodes are mobile and relays are allowed to carry buffered packets while they move. Both strategies are based on the following model: if $p_n$ is the transmission rate of each node, {\it i.e.}, the proportion of time each node is active and transmitting, the radius of efficient transmission is given by $r_n\sim L\sqrt{\frac{\kappa}{np_n}}$ when $n$ approaches infinity for some constant $\kappa>0$ which depends on the protocol, interference model, etc. 

In the context of~\cite{gk2000}, the number of relays a packet has to traverse to reach its destination is \mbox{$h_n = O(1/r_n)$}. Consequently, $np_nC$ must be divided by $h_n$ to get the useful capacity: \mbox{$np_nC/h_n=O(C\sqrt{p_n n})$}. In order to ensure connectivity in the network, so that every source is able to communicate with its randomly chosen destination, $p_n$ must satisfy the limit \mbox{$p_n\le O(1/\log n)$}. This leads to Gupta and Kumar's maximum capacity of $O(C\sqrt{n/\log n})$ with ``hot potatoes'' routing. 

In contrast, in the context of~\cite{gt2002}, the network does not need to be connected since the packets are mostly carried in the buffer of a mobile relay. Therefore there is no limit on $p_n$ other than the requirement that it must be smaller than some $\varepsilon<1$ that depends on the protocol and some other physical parameters. Thus $r_n$ is $O(1/\sqrt{n})$. In Grossglauser and Tse's model, the source transmits the packet to the closest mobile relay or keeps it until it finds one. This mobile relay delivers the packet to the destination when it comes within range of the destination node. Such a packet delivery requires a transmission phase which also includes retries and acknowledgements so that the packet delivery can be eventually guaranteed. 

The proposed model of~\cite{gt2002} requires a GPS-like positioning system and the knowledge of the effective range $r_n$. The estimate of $r_n$ could be achieved via a periodic beaconing from every node, where each beacon contains the position coordinates of the node, so that a node knows the typical distance for a successful reception. However, the relay cannot rely on beaconing in order to detect when it is in the reception range of the destination. The reason is that a node stays in the reception range of another node for a short time period of order \mbox{$r_n/v=1/\sqrt{n}$} and this cannot be detected via a periodic beaconing with bounded frequency since $p_n=O(1)$ (the frequency of periodic beaconing should be of $O(\sqrt{n})$). We may also assume that the destination node is fixed and its cartesian coordinates are known by the mobile relay. Otherwise, if the destination node is mobile, there would be a requirement for this node to track its new coordinates and disseminate this information in the wireless network as in \cite{homeagent,Li:2000}. 

It is also interesting to note that Diggavi, Grossglauser, and Tse~\cite{Diggavi05} showed that a constant throughput per source-destination pair is feasible even with a more restricted mobility model. Franceschetti et al.~\cite{Franceschetti07} proved that there is no gap between the capacity of randomly located and arbitrarily located nodes. Throughput and delay trade-offs have appeared in~\cite{ElGamal04,Sharma07} where delay of multi-hop routing is reduced by increasing the coverage radius of each transmission, at the expense of reducing the number of simultaneous transmissions the network can support.  We will show that, in our work, we have a delay of $O(1/v)$ and throughput per source-destination pair of $O(\frac{1}{\log n\log\log n})$. If we take the notation of $\sqrt{a(n)}$ in~\cite{ElGamal04,Sharma07} to measure the average distance traveled toward the destination between two consecutive emissions of the same packet, then we will show that our scheme yields $\sqrt{a(n)}=\Theta(1/\log n)$. If we compare with the result of~\cite{ElGamal04,Sharma07}, we should have a throughput of $\Theta(\frac{1}{\log n\sqrt{n\log n}})$ but our scheme delivers a higher throughput by a factor greater than $\sqrt{n}$. In fact, if $\ell$ is the average free space distance of the random walk, then our scheme yields $\sqrt{a(n)}=\Theta\left(\frac{1}{\frac{1}{\ell}+\log n}\right)$. The apparent contradiction comes from the fact that the authors in~\cite{ElGamal04,Sharma07} consider a mobility model based on brownian motion. This corresponds to having $\ell=0$ and, in this case, our scheme would be equivalent to the ``hot potatoes'' routing of \cite{gk2000} with \mbox{$\sqrt{a(n)}=\Theta(r_n)$}. Let us point out that the brownian motion mobility is an important yet worst case model and it is not realistic for real world situations such as urban area mobile networks. In the section devoted to generalizations, we extend our result to fit a more general mobility model where mobile nodes follow fractal trajectories with $\ell=\ell_n=\Theta(1/\log n)$ and the throughput of our scheme remains of $\Theta(\frac{1}{\log n\log\log n})$.

On the practical side, many protocols have been proposed for wireless multi-hop networks. These protocols may be classified in topology-based and position-based protocols. Topology-based protocols~\cite{OLSR,ZRP,AODV} need to maintain information on routes potentially or currently in use, so they do not work effectively in environments with high frequency of topology changes. For this reason, there has been an increasing interest in position-based routing protocols. In these protocols, a node needs to know its own position, the one-hop neighbors' positions, and the destination node's position. These protocols do not need control packets to maintain link states or to update routing tables. Examples of such protocols can be found in \cite{Navas97,Karp00,Ko00,Basagni98,1096061,1096436,Kranakis99compassrouting,1096124}. In contrast to our work, they do not analyze the trade-off between the capacity and the delay of the network under these protocols and their scaling properties. 

\section{Network and Mobility Settings}
\label{sec:model}

We consider a network of $n$ mobile nodes with their initial positions uniformly distributed over the network area. Each mobile node transmits packets to a randomly chosen fixed node, called its destination node, which is also randomly located in the network area. We assume that mobile nodes are aware of their own cartesian coordinates, {\it e.g.}, by using GPS or from the initial position, a mobile node could use the knowledge of its motion vector to compute its cartesian coordinates at any given time.  

Initially we consider that only mobile nodes participate in the relay process to deliver packet to its destination node. The case where the fixed nodes may also participate in the relay process is discussed in Section \ref{sec:extension}. A mobile node should be aware of the cartesian coordinates of the destination node of a packet it carries. Indeed it can be assumed that this information is included in all packets or is relayed with the packets. Hence our model only requires that a source or relay node is aware of the cartesian coordinates of the destination node which is assumed fixed. Note that if the destination node is mobile, a mechanism to disseminate its updated cartesian coordinates in the network can be used, {\it e.g.}, \cite{homeagent,Li:2000}. However, this is outside the scope of this paper as we particularly focus on the throughput-delay tradeoff.

 With the available information, a mobile relay can determine:
\begin{compactitem}[-]
\item its heading vector, which is the motion vector when its speed is non zero,
\item its bearing vector, which is the vector between its position and the position of a packet's destination; and,
\item the relative bearing angle, {\it i.e.}, the absolute angle between its heading and bearing vectors. 
\end{compactitem}
In the example of Fig.~\ref{fig:protocol}, node $A$ is carrying a packet for node $D$. This figure also shows the heading vector of mobile node $A$ and its bearing vector and relative bearing angle for destination node $D$. Note that a mobile relay may carry packets for multiple destinations but can easily determine the bearing vector and relative bearing angle for each destination node. 

\section{Model of CRB Scheme}
\label{sec:protocol}

In this section, we will present the parameters and specifications of the model of our georouting scheme.

\subsection{Parameters}

We define the parameters $\theta_c$, called the carry angle, and $\theta_e$, called the emission angle. Each mobile node carries a packet to its destination node as long as its relative bearing angle, $\theta$, is smaller than $\theta_c$ which is strictly smaller than $\pi/2$. When this condition is not satisfied, the packet is transmitted to the next relay. 

\subsection{Model Specification With Radio Range Awareness}
\label{sec:specdiskgraph}

In the following description, we initially assume that each node is aware of the effective range of transmission $r_n$. This means that there is a periodic beaconing that allows this estimate to be made. In Section \ref{subsec:withoutawareness}, we will investigate how to specify our model without an estimate of the effective range $r_{n}$. 

Assume that node $A$ is carrying a packet for node $D$. The velocity of node $A$ is denoted by $\bv(A)$. 
\begin{compactitem}[-]
\item If node $A$ is within range of node $D$, it transmits the packet to $D$; otherwise,
\item if the relative bearing angle is smaller than $\theta_c$, node $A$ continues to carry the packet; otherwise, 
\item node $A$ transmits the packet to a random neighbor mobile node inside the cone of angle $\theta_e$, with bearing vector as the axis, and then forgets the packet.
\end{compactitem}

\begin{figure} [!t]
\centering
\psfrag{a}{$\theta$ (relative bearing angle)}
\psfrag{b}{$\theta_e$}
\psfrag{c}{Mobile node `A'}
\psfrag{d}{Destination node `D'}
\psfrag{e}{Bearing vector}
\psfrag{f}{Heading vector}
\includegraphics[scale=0.55]{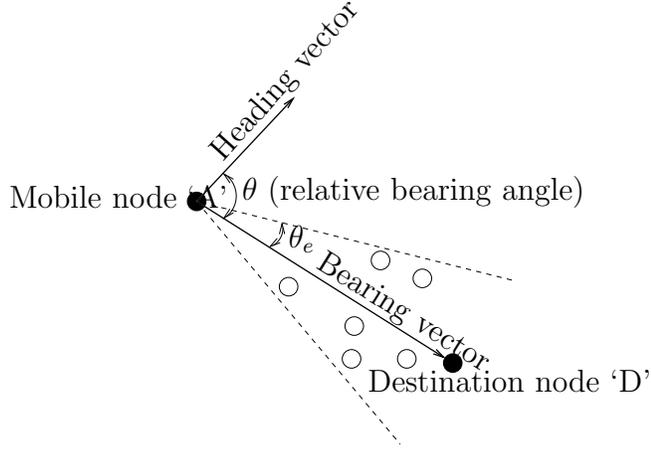}
\caption{Figurative representation of our model. Unfilled circles represent the potential mobile relays for packet transmitted by node $A$ for node $D$.
\label{fig:protocol}}
\end{figure}

In order to better understand the model of our georouting scheme, consider the example in Fig. \ref{fig:protocol}. Assume that node $A$ is out of range of node $D$ and, because of that, it cannot deliver the packet directly. Now, if $\theta<\theta_c$, node $A$ will continue to carry the packet for node $D$. Otherwise, it transmits the packet to one of the random mobile relays, represented by unfilled circles in the figure. 

  

\subsubsection{Transmission procedure}
\label{sec:tx_procedurediskgraph}

To transmit the packet towards another mobile node, node $A$ shall proceed as follows:
\begin{compactitem}[-]
\item it first transmits a {\em Call-to-Receive} packet containing the positions of nodes $A$ and $D$;
\item a random mobile node $B$  which receives this {\em Call-to-Receive} packet can compute the angle $(AB,AD)$. If this angle is smaller than $\theta_e$, it replies with an {\em Accept-to-Receive} packet containing an identifier of node $B$;
\item node $A$ sends the packet to the first mobile node which replied with an {\em Accept-to-Receive} packet.
\end{compactitem}
The first node which sends its {\em Accept-to-Receive} packet notifies the other receivers of the {\em Call-to-Receive} packet, to cancel their transmissions of {\em Accept-to-Receive} packets. There may be more than one (but finite) transmissions of {\em Accept-to-Receive} packets in case two or more receivers are at distance greater than $r_{n}$ from each other.

Note that this procedure does not need any beaconing or periodic transmission of hello packets. The back-off time of nodes, transmitting their {\em Accept-to-Receive} packet, can also be tuned in order to favor the distance or displacement towards $D$, depending on any additional optional specifications.

\subsection{Model Specification Without Radio Range Awareness}
\label{subsec:withoutawareness}

The estimation of $r_{n}$ would require that the nodes employ a periodic beaconing mechanism. If such a mechanism is not available or desirable, the CRB scheme relies on the signal to interference plus noise ratio (SINR) for transmitting packets to their destinations or random mobile relays. In other words, a mobile node can relay a packet to its destination node or another mobile node only if the SINR at the receiver is above a given threshold. 

Note that in this case, the specification of the transmission procedure is also modified so that it terminates when the final destination receives the packet. To transmit the packet towards its destination node or another mobile node, node $A$ shall proceed as follows: 
\begin{compactitem}[-]
\item it first transmits a {\em Call-to-Receive} packet containing the positions of nodes $A$ and $D$;
\item if node $D$ receives this packet, it responds immediately with an {\em Accept-to-Receive} packet with highest priority. Node $A$, on receiving this packet, relays the packet to node $D$; otherwise, 
\item the procedure of selecting a random mobile node, as the next relay, is similar to the procedure described in Section \ref{sec:tx_procedurediskgraph}. A random mobile node $B$, which receives the {\em Call-to-Receive} packet, computes the angle $(AB,AD)$. If this angle is smaller than $\theta_e$, it responds with an {\em Accept-to-Receive} packet; 
\item node $A$ relays the packet to the first mobile node which sent its {\em Accept-to-Receive} packet successfully. The first node which transmits its {\em Accept-to-Receive} packet also makes the other receivers to cancel their transmissions of {\em Accept-to-Receive} packets. 
\end{compactitem}

\section{Performance analysis}
\label{sec:analysis}

We will show that our georouting scheme is stable as long as the average transmission rate of each mobile node is \mbox{$p_n=O(1/\log\log n)$}. We will also show that the number of transmissions per packet is of $O(\log n)$ and this would lead to a useful network capacity of $O(C\frac{n}{\log n\log\log n})$.

We assume that the network area is a square area and without loss of generality we assume that it is a square {\sl unit} area. The mobile nodes move according to i.i.d. random walk: from a uniformly distributed initial position, the nodes move in a straight line with a certain speed and randomly change direction. The speed is randomly distributed in an interval $[\vmin,\vmax]$ with \mbox{$\vmin>0$}. To simplify the analysis, we assume that \mbox{$\vmin=\vmax=v$}. We also assume that each node changes its direction with a Poisson point process of rate $\tau$. When a mobile node hits the border of the network, it simply bounces like a billiard ball. This leads to the {\it isotropic property}~(Jacquet et al.~\cite{jmr2010}): at any given time the distribution of mobile nodes is uniform on the square and the speed are uniformly distributed in direction independently of the position in the square.

We assume that the radius $r_n$ of efficient transmission is given by 
$$
r_n=\sqrt{\beta\frac{\log\log n}{\pi n}}~,
$$ 
for some $\beta>0$. Therefore, the average number of neighbors of an arbitrary node at an arbitrary time is $\beta\log\log n$.
In order to keep the average cumulated load finite, the nodes have an average transmission rate of \mbox{$p_n=\frac{1}{\beta\log\log n}$}. 
Therefore, the actual density of simultaneous transmitters is $\frac{n}{\beta\log\log n}$.

\subsection{Methodology}
The parameters of interest are the following:
\begin{compactitem}[-]
\item The delay $D_n(r)$ of delivering a packet to the destination when the packet is generated in a mobile node at distance $r$ from its destination node.
\item The average number of times $F_n(r)$ the packet changes relay before reaching its destination when it has been generated in a mobile node at distance $r$ from its destination node.
\end{compactitem}
In order to exhibit the actual performance of our proposed CRB scheme, we aim to derive an upper-bound on the parameters $D_n(r)$ and $F_n(r)$. In the next two sub-sections, we assume w.l.o.g. that there is always a relay node, to receive the packet, in the emission cone (as the node density and angle, $\theta_e$, are sufficiently large) when a relay change must occur.

\subsection{Delivery Delay}
In the quantity $D_n(r)$, we ignore the queueing delay which can become apparent when several packets could be in competition in the same relay to be transmitted at the same time. We analyze the delay under the hypothesis that store and forward delays are negligible (these delays would be negligible as long as the queue length is bounded).
\begin{theorem}
We have the bound 
\begin{equation}
D_n(r)\le\frac{r}{v\cos(\theta_c)}~.
\label{eq:delay}
\end{equation}
\end{theorem}
\begin{proof}
During a relay change, the new relay is closer to the destination than the previous relay. Ignoring relay changes that take zero time, and neglecting the distance decrement during relay change, the packet moves at constant speed $v$ with a relative bearing angle always smaller than $\theta_c$.
\end{proof}

\subsection{Number of Relay Changes}
There are two events that trigger relay changes. 
\begin{compactenum}
\item Relay change due to turn, {\it i.e.}, the mobile node, carrying the packet, changes its heading vector such that the relative bearing angle becomes greater than $\theta_c$.
\item Relay change due to pass over, {\it i.e.}, the mobile node keeps its trajectory and the relative bearing angle becomes greater than $\theta_c$. 
\end{compactenum}

Consider a packet generated at distance $r$ from its destination. Let $F_n^t(r)$ be the average number of relay changes due to turn. Equivalently, let $F_n^p(r)$ be the average number of relay changes due to pass over. Therefore, we have $F_n(r)=F_n^t(r)+F_n^p(r)$ and we expect that the main contribution of $O(\log n)$ in $F_n(r)$ will come from $F_n^p(r)$. 

\subsubsection{Number of Relay Changes Due to Turn}
We prove the following theorem:
\begin{theorem}
We have the bound
\begin{equation*}
F_n^t(r)\le\frac{\pi-\theta_c}{\theta_c}\frac{\tau}{v\cos(\theta_c)}r~.
\end{equation*}
\end{theorem}
\begin{proof}

\begin{figure} [!t]
\centering
\psfrag{a}{$\theta<\theta_c$}
\psfrag{b}{$\theta'>\theta_c$}
\psfrag{r}{$r$}
\includegraphics[scale=0.6]{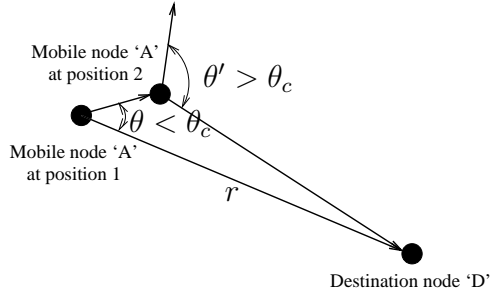}
\caption{Figurative description of relay change due to turn. At position $1$, $\theta<\theta_c$ and node $A$ carries the packet for node $D$. At position $2$, node $A$ changes its heading vector and must transmit the packet.
\label{fig:case1}}
\end{figure}

We consider the case in Fig. \ref{fig:case1} and assume that a mobile node is carrying a packet to its destination located at distance $r$. 
The node changes its direction with Poisson rate $\tau$. When the node changes its direction, it may keep a direction that stays within angle $\theta_c$ with the bearing vector and this will not trigger a relay change. This occurs with probability $\frac{\theta_c}{\pi}$. Otherwise, the packet must change relay. But the new relay may have relative bearing angle greater than $\theta_c$ which would result in an immediate new relay change. Therefore, at each direction change, there is an average of $\frac{\pi-\theta_c}{\theta_c}$ relays. Multiplied by $D_n(r)$ this gives our upper-bound of $F_n^t(r)$. 

Note that we have not considered the turn due to bounces on the borders of square map. But it is easy to see via straightforward geometric considerations that they cannot actually generate a relay change. 
\end{proof}

\subsubsection{Number of Relay Changes Due to Pass Over}
We prove the following theorem:
\begin{theorem}
We have the bound
\begin{equation*}
F_n^p(r)\le\frac{\pi\tan(\theta_c)}{\theta_c^2}\log\left(\frac{r}{r_n}\right)~.
\end{equation*}
\end{theorem}
\begin{proof}
Here we consider the case of Fig. \ref{fig:case2}. We assume that a mobile node at distance $r$, from its destination, has a relative bearing angle equal to $\theta$. If it keeps its trajectory ({\it i.e.,} does not turn), it will need to transmit to a new relay when it passes over the destination, {\it i.e.}, when it arrives at a distance of \mbox{$\rho(\theta,r)=\frac{\sin(\theta)}{\sin(\theta_c)}r$} from the destination. The function of $\theta$ $\rho(\theta,r)$ is bijective from $[0,\theta_c]$ to $[0,r]$. For $x\in[0,r]$ let  $\rho^{-1}(x,r)$ be its inverse.

\begin{figure} [!t]
\centering
\psfrag{a}{$\theta<\theta_c$}
\psfrag{b}{$\theta'=\theta_c$}
\psfrag{r}{$r$}
\psfrag{d}{$\rho(\theta,r)$}
\includegraphics[scale=0.6]{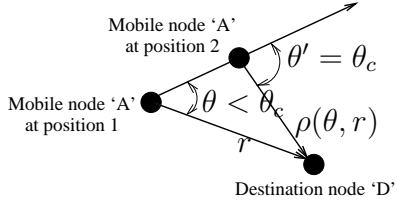}
\caption{Figurative description of relay change due to pass over. At position $1$, $\theta<\theta_c$ and node $A$ carries the packet for node $D$. At position $2$, node $A$ has the same heading vector but $\theta'=\theta_c$ and it must transmit the packet.
\label{fig:case2}}
\end{figure}

Assume that $r$ is the distance to the destination when the relay receives the packet or just after a turn. Thus the angle $\theta$ is uniformly distributed on $[0,\theta_c]$, {\it i.e.}, with a constant probability density $\frac{1}{\theta_c}$.
The probability density of the pass over event at $x<r$ (assuming no direction change) is therefore
$$
\frac{1}{\theta_c}\frac{\partial}{\partial x}\rho^{-1}(x,r)=\frac{\sin(\theta_c)}{\theta_c\cos(\rho^{-1}(x,r))r}=\frac{\tan(\rho^{-1}(x,r))}{\theta_c}\frac{1}{x}~.
$$


Since $\rho^{-1}(x,r)\le\theta_c$, the point process where the packet would need a relay change due to pass over is upper bounded by a Poisson point process on the interval $[r_n,r]$ and of intensity equal to $\frac{\tan(\theta_c)}{\theta_c}\frac{1}{x}$ for $x\in[r_n,r]$. 

Since a relay change due to pass over corresponds to an average of $\frac{\pi}{\theta_c}$ relays, and neglecting the decrement of distance during each transmission phase, we get
\begin{equation*}
F_n^p(r)=\int_{r_n}^r\frac{\pi\tan(\theta_c)}{\theta^2_c}\frac{dx}{x}=\frac{\pi\tan(\theta_c)}{\theta^2_c}\log \left(\frac{r}{r_n}\right)~.
\end{equation*}
\end{proof}
We have thus
\begin{equation*}
F_n(r)\le\frac{\pi-\theta_c}{\theta_c}\frac{\tau}{\cos(\theta_c)v}r+\frac{\pi\tan(\theta_c)}{\theta^2_c}\log \left(\frac{r}{r_n}\right)~.
\end{equation*}

Therefore we have a main contribution of $O(\log n)$ relay changes that comes from $\log(1/r_n)$. The result holds because we assume that there is always a receiver in each relay change. In the next sub-section we remove this condition to establish a result with high probability. 

\subsection{Number of Relay Changes With High Probability of Success}
In the previous subsection we assumed that there is always a receiving relay in the emission cone at each relay change and we said that the relay change is always successful. The case with failed relay change would introduce additional complications. For example one could use the fixed relays if the packet cannot be delivered to a mobile relay. Anyhow, to simplify the present contribution, we will show that with high probability, {\it i.e.}, with probability approaching one when $n$ approaches infinity, every relay change succeeds.
\begin{theorem}
With high probability on arbitrary packets, all relay changes succeed for this packet and are in average number $F_n(r)$ and the delay is $D_n(r)$.
\end{theorem}
\begin{proof}
We use a modified stochastic system to cope with failed relay changes. The modification is the following: when there is no relay in the emission cone during a relay change a {\em decoy} mobile relay is created in the emission cone that will receive the packet. Each decoy relay is used only for one packet and disappear after use. Notice that the modified system is {\it not} a practical scheme in a practical network. The analysis in the previous section still holds and in particular $F_n(r)$ is now the average unconditional number of relay changes (including those via decoy relays) for any packet starting at distance $r$ from destination. 

Let $P_n(r)$ be the probability that a packet starting at distance $r$ has a failed relay change. The probability that a relay change fails is equal to $(1-\theta_e r_n^2)^{n-1}\sim e^{-n\theta_e^2r_n^2}=(\log n)^{-\beta\frac{\theta_e}{\pi}}$. Therefore the average number of failed relay changes $E_n(r)\le F_n(r)(\log n)^{-\beta\frac{\theta_e}{\pi}}$ which tends to zero when $\beta\frac{\theta_e}{\pi}>1$, since $F_n(r)=O(\log n)$. The final result comes since $P_n(r)\le E_n(r)$.
\end{proof}

\section{Simulations}
\label{sec:simulations}
We performed simulations with CRB georouting scheme under two contexts:
\begin{compactenum}
\item a simplified context where the network is modeled under unit disk model;
\item a realistic context where the network operates under slotted ALOHA and a realistic SINR interference model is considered. The simulations of CRB scheme are stressed to the point that the motion timings are not so large compared to slot times. 
\end{compactenum}
\subsection{Under Disk Graph Model}
\label{sec:simdiscunit}

In this section, we consider a network of $n$ mobile nodes. We assume that all nodes have the same radio range given by
\[
r_n=\sqrt{\frac{\beta_0\log\log n}{\pi n}}~.
\]
Each mobile node moves according to an i.i.d. random walk mobility model, {\it i.e.}, it starts from a uniformly distributed initial position, moves in straight line with constant speed and uniformly selected direction and reflects on the borders of the square area (like billiard balls). 

In the next section (Section~\ref{sec:protocolsinr}), we will further explore the effect of interference on the simulations, but for the moment we only consider a source mobile node and its randomly located destination node which is fixed. We adopt the disk graph model of interference, {\it i.e.}, two nodes are connected or they can exchange information if the distance between them is smaller than a certain threshold (called radio range), otherwise, they are disconnected. A mobile node relays the packet only if the relative bearing angle, {\it i.e.}, the absolute angle made by the heading vector and the bearing vector, becomes greater than~$\theta_c$. Otherwise, it continues to carry the packet.

\subsubsection{Simulation parameters and assumptions}

The purpose of our simulations is to verify the scaling behavior of average delay and number of hops per packet with increasing number of nodes in the network. Therefore, the number of mobile nodes, $n$, in the network is varied from $10000$ to two million nodes. The values of other parameters, which remain constant, and do not impact the scaling behavior are listed as follows.
\begin{compactenum}[a]
\item[{\it (i)}] Parameters of CRB scheme, $\theta_c$ and $\theta_e$, are taken to be $\pi/6$.
\item[{\it (ii)}] The speed of all mobile nodes is constant, {\it i.e.}, $0.005$ unit distance per slot.
\item[{\it (iii)}] All mobile nodes change their direction according to a Poisson point process with mean equal to $10$ slots.
\item[{\it (iv)}] The value of constant factor $\beta_0$ is assumed to be equal to~$40$.
\end{compactenum}

\subsubsection{Results}

We have evaluated the following parameters.
\begin{compactenum}[a]
\item[{\it (i)}] Average delay per packet.
\item[{\it (ii)}] Average number of hops per packet.
\end{compactenum}

We considered the Monte Carlo Method with $100$ simulations. The delay of a packet is computed from the time when its processing started at its source mobile node until it reaches its destination node. Figure \ref{fig:delayvsnumberofnodes} shows the average delay per packet with an increasing number of nodes. We notice that as $n$ increases, the average delay per packet appears to approach a constant upper bound which can be computed from \eqref{eq:delay}. Figure \ref{fig:transmissionsvsnumberofnodes} shows the average number of hops per packet with increasing values of $n$.

\begin{figure}[t!]
\centering
\includegraphics[scale=0.77]{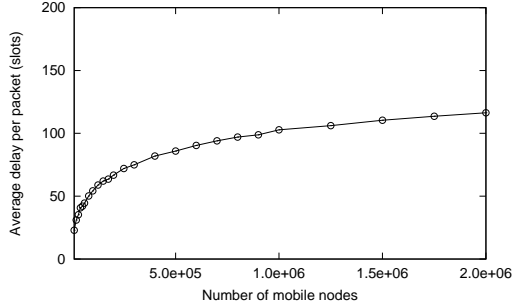}
\caption{Average delay per packet.}
\label{fig:delayvsnumberofnodes}
\end{figure}

\begin{figure}[t!]
\centering
\includegraphics[scale=0.77]{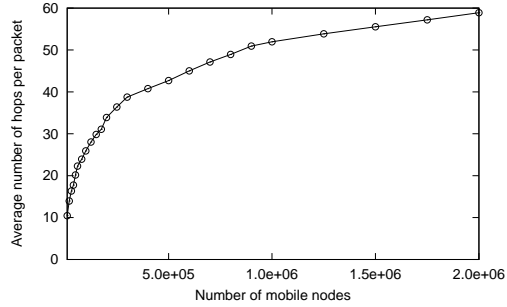}
\caption{Average number of hops per packet.}
\label{fig:transmissionsvsnumberofnodes}
\end{figure}

\subsection{With slotted ALOHA under SINR interference model}
\label{sec:protocolsinr}

In this section, we will present the simulations of CRB georouting scheme with a transmission model which does not rely on the estimate of $r_{n}$ and is based on the required minimal SINR threshold. 

\subsubsection{Transmission model}
Our transmission model is as follows. Let $P_i$ be the transmit power of node $i$ and $\gamma_{ij}$ be the channel gain from node $i$ to node $j$ such that the received power at node $j$ is $P_i\gamma_{ij}$. The transmission from node $i$ to node $j$ is successful only if the following condition is satisfied
$$
\frac{P_i\gamma_{ij}}{N_0+\sum_{k\neq i}P_k\gamma_{kj}}>K~,
$$
where $K$ is the desired minimum SINR threshold for successfully receiving the packet at the destination and $N_0$ is the background noise power. For now, we ignore multi-path fading or shadowing effects and assume that the channel gain from node $i$ to node $j$ is given by 
$$
\gamma_{ij}=\frac{1}{|z_i-z_j|^{\alpha}}~,
$$
where $\alpha>2$ is the attenuation coefficient and $z_i$ is the location of node $i$.

\subsubsection{Simulations under SINR interference model}
\label{sec:simunderSINR}

For the theoretical analysis in Section \ref{sec:analysis}, we have assumed that the effective range of successful transmission is 
$$
r_n = \sqrt{\beta\frac{\log \log n}{\pi n}},
$$
which requires that the mobile nodes have an average transmission rate of $p_n=\beta/\log \log n$. In other words, if the mobile nodes emit packets at the given average rate, the average distance of successful transmission under SINR interference model is of $O(r_n)$ and the results from theoretical analysis are applicable as well. 

We assume that time is slotted and mobile nodes determine their relative bearing angles at the beginning of a slot. We also assume that all nodes are synchronized and simultaneous transmitters in each slot emit a {\em Call-to-Receive} packet at the beginning of the slot. Moreover, we also assume that fixed nodes do not emit any packet except, maybe, an {\em Accept-to-Receive} packet in response to a transmission by a mobile node. 

In our simulation environment, $n$ mobile nodes start from a uniformly distributed initial position and move independently in straight lines and in randomly selected directions. They also change their direction randomly at a rate which is a Poisson point process. Each mobile node sends packets towards a unique destination (fixed) node, and all destinations nodes are also uniformly distributed in the network area.  

In order to keep load in the network finite, the packet generation rate at a node, $\rho_n$, should be of $O(p_n/X_{n})$ where $X_n$ is the expected number of transmissions per packet. From the theoretical analysis, we know that
$$
X_n=O\left(\log\left(\frac{n}{\beta_2}\right)\right) + c~,
$$ 
where $c$ is a constant if $\theta_c$ is non-varying. In our simulations under SINR interference model, we assume that the knowledge of $r_n$ is not available and mobile nodes use minimal SINR threshold for successfully receiving a packet. We also assume that each mobile node generate packets, destined for its unique fixed destination node, at a uniform rate given by
\begin{equation}
\rho_n=\frac{1}{\beta_1\log(\frac{n}{\beta_2})\log \log n}~,
\label{eq:rho}
\end{equation}
for some \mbox{$\beta_1>0$} and \mbox{$\beta_2>0$}. 
We ignored the value of constant $c$ and have observed that the simulation results are asymptotically correct because, with $n$ increasing, value of $c$ should be insignificant as compared to the $O(\log(n/\beta_2))$ factor.

\subsubsection{Simulation parameters and assumptions}

The purpose of our simulations is to verify the scaling properties of network capacity, delay and number of transmissions per packet with increasing number of nodes in the network. The number of mobile nodes, $n$, in the network is varied from $250$ nodes to 100,000 nodes. All nodes use uniform unit nominal transmit power and the background noise power $N_0$ is assumed to be negligible. The values of other parameters are listed as follows.

\begin{compactenum}[a]
\item[{\it (i)}] Parameters of CRB scheme, $\theta_c$ and $\theta_e$, are taken to be $\pi/6$.
\item[{\it (ii)}] The speed of all mobile nodes is constant, {\it i.e.}, $0.01$ unit distance per slot.
\item[{\it (iii)}] All mobile nodes change their direction independently and randomly according to a Poisson point process with mean equal to $10$ slots.
\item[{\it (iv)}] The values of constant factors $\beta_1$ and $\beta_2$ are assumed equal to $500$ and $1$ respectively.
\item[{\it (v)}] SINR threshold, $K$, is assumed equal to $1$.
\item[{\it (vi)}] Attenuation coefficient, $\alpha$, is assumed equal to $2.5$.
\end{compactenum}

In our simulations, we make the following assumptions. 

\begin{compactenum}[a]
\item[{\it (i)}] Each mobile node generates an infinite number of packets, at rate $\rho_n$, for its respective destination node. 
\item[{\it (ii)}] A mobile node may carry, in its buffer, its own packets as well as the packets relayed from other mobile nodes. Therefore, it may have more than one packet in its buffer which it must transmit because their respective relative bearing angles are greater than $\theta_c$. In such a case, it first transmits the packet which is furthest from its destination.
\end{compactenum}

\subsubsection{Results}

We have examined the following parameters.

\begin{compactenum}[a.]
\item[{\it (i)}] Throughput capacity per node, $\lambda_{n}$. 
\item[{\it (ii)}] Average number of hops, $h_{n}$, and transmission attempts, $t_{n}$, per packet.
\item[{\it (iii)}] Average delay per packet.
\end{compactenum}

The throughput capacity per node, $\lambda_n$, is the average number of packets arriving at their destinations per slot per mobile node. With $n$ increasing, throughput capacity per node should follow the following relation
\begin{equation}
\lambda_n=\frac{\eta}{{\beta_1}{\log(\frac{n}{\beta_2})\log \log n}}~,
\label{eq:lambda}
\end{equation}
for some \mbox{$0<\eta<1$} which depends on $K, \alpha$ and protocol parameters. Note that the values of these constants do not affect the asymptotic behavior of $\lambda_{n}$ which is also observed in our simulation results. 

In order to verify the asymptotic character of simulated packet generation rate and throughput capacity, we have analyzed the parameters $m_{\rho}$ and $m_{\lambda}$ which are given by 
\begin{eqnarray*}
m_{\rho}&=&\rho_n\left(\beta_1\log\left(\frac{n}{\beta_2}\right)\log \log(n)\right)~,\\
m_{\lambda}&=&\lambda_n\left(\beta_1\log\left(\frac{n}{\beta_2}\right)\log \log(n)\right).
\end{eqnarray*}
From the definition of $\rho_n$ in \eqref{eq:rho}, the value of $m_{\rho}$ should be constant at $1$ whereas, with $n$ increasing, value of $m_{\lambda}$ should converge to the constant $\eta$. From Fig. \ref{fig:net_capacity_verify}, value of $\eta$ is found to be approximately equal to $0.45$. Figure \ref{fig:net_capacity} shows the simulated and theoretical packet generation rate, $n\rho_{n}$, and throughput capacity, $n\lambda_{n}$, in the network.  The theoretical values of $n\rho_n$ and $n\lambda_n$ are computed from \eqref{eq:rho} and \eqref{eq:lambda}.  

\begin{figure} [!t]
\centering
\includegraphics[scale=0.77]{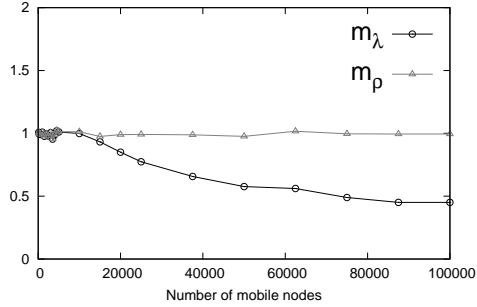}
\caption{Verification of network throughput capacity with plots of $m_{\lambda}$ and $m_{\rho}$.
\label{fig:net_capacity_verify}}
\end{figure}

\begin{figure} [!t]
\centering
\includegraphics[scale=0.77]{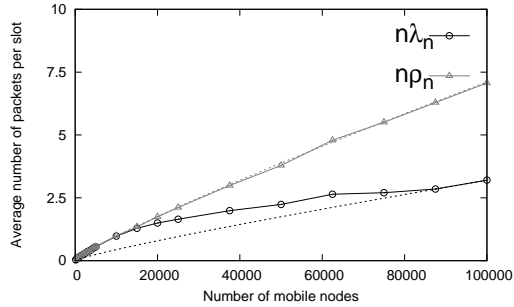}
\caption{Simulated (solid lines) and theoretical (dotted lines) network throughput capacity, $n\lambda_n$, and network packet generation rate, $n\rho_n$.\label{fig:net_capacity}}
\end{figure}

\begin{figure} [!t]
\centering
\includegraphics[scale=0.77]{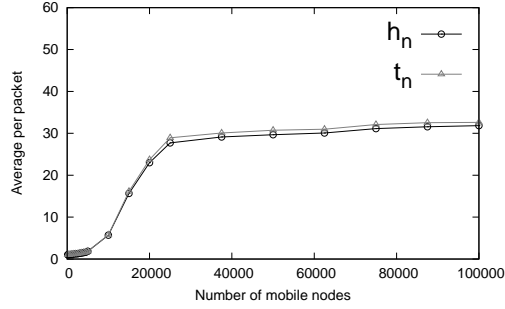}
\caption{Average number of hops, $h_n$, and transmission attempts, $t_n$, per packet.\label{fig:hops_tx}}
\end{figure}

\begin{figure} [!t]
\centering
\includegraphics[scale=0.77]{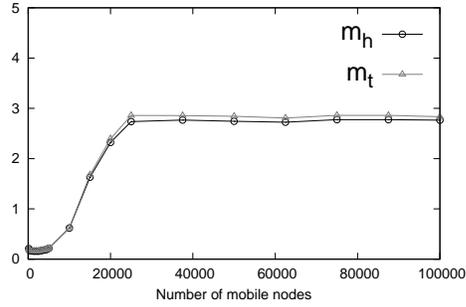}
\caption{Verification of number of hops and transmission attempts with plots of $m_h$ and $m_t$.
\label{fig:hops_tx_verify}}
\end{figure}

\begin{figure} [!t]
\centering
\includegraphics[scale=0.77]{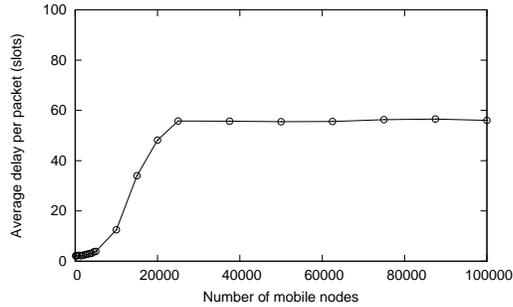}
\caption{Average delay per packet.\label{fig:delay}}
\end{figure}

Figure \ref{fig:hops_tx} shows the average number of hops, $h_n$, and transmission attempts, $t_n$, per packet. The value of $t_n$ is slightly higher than the value of $h_n$ because of the possibility that a successful receiver may not be found in each transmission phase, {\it i.e.}, in the cone of transmission formed with $\theta_e$. With $n$ increasing, $h_n$ and $t_{n}$ are expected to grow in $O(\log (n/\beta_2))$. To verify this character in simulation results, we examine the parameters $m_{h}$ and $m_{t}$ given by
\begin{eqnarray*}
m_{h}&=&h_n\frac{1}{\log(\frac{n}{\beta_2})}~,\\
m_{t}&=&t_n\frac{1}{\log(\frac{n}{\beta_2})}.
\end{eqnarray*}
If the values of $h_{n}$ and $t_{n}$ are in $O(\log (n/\beta_2))$, the values of $m_{h}$ and $m_{t}$ should approach a constant value which is the case in Fig.~\ref{fig:hops_tx_verify}.

The delay of a packet is computed from the time when its processing started at its source mobile node until the time it arrives at its destination node. Figure~\ref{fig:delay} shows the average delay per packet. As the number of mobile nodes increase, the average delay appears to approach a constant value.

It can be observed that when $n$ is small, the average number of hops per packet is almost of $O(1)$ which also means that the average delay per packet is of $O(1)$ and the network throughput capacity is of $O(\eta n)$: although, in simulation results, it is bounded by the network packet generation rate which is of $O(\frac{n}{\log n \log\log n})$. This can be observed in Fig. \ref{fig:net_capacity}, \ref{fig:hops_tx} and \ref{fig:delay}. The reason is that when $n$ is small, the number of simultaneous transmissions in the network is also small and packets can be delivered by the mobile nodes, directly to their destination nodes, in $O(1)$ hops. As $n$ increases, number of simultaneous transmitters increase and consequently the effective transmission range of each transmitter shrinks. Therefore, the dominant factor in the number of transmissions per packet comes from the fact that a mobile relay has to be close to the destination, to deliver a packet. According to theoretical analysis, $h_n$ and $t_n$ grow in $O(\log n)$ which is also observed in the simulation results. Simulations also show that, asymptotically, network throughput capacity is of $O(\frac{n}{\log n \log\log n})$ and average delay per packet is of $O(1/v)$ which complies with our theoretical analysis.

\section{Extensions and general mobility models}
\label{sec:extension}

In our discussion, we primarily focussed on the capacity-delay tradeoff and thus for the initial sake of clarity assumed that the fixed nodes can only receive packets destined for them.  We could also consider a slight variation in the specification of the model of CRB scheme such that the fixed nodes also participate in the routing of packets to their destination nodes. For example, during a transmission phase, if a packet cannot be transmitted to its destination node or relayed to a random mobile neighbor in the cone of transmission, it can be relayed to a fixed node. This fixed node must emit this packet immediately to its destination node or to any mobile relay in the neighborhood. Note that this will also help increase the connectivity of the network. 

The condition about i.i.d. random walks can be relaxed and the result about the expected number of relay changes will still be valid. In other words, the i.i.d. random walk model can be seen as a worst case compared to realistic mobility models. If the mobile relays move like cars in an urban area, then  we can expect that their mobility model will significantly depart from the random walk. Indeed cars move toward physical destinations and in their journey on the streets toward their destination, their heading after each turn is positively correlated with the heading before the turn. This implies that the probability that a relay change is needed after a turn is smaller than it would be under a random walk model, where headings before  and after turn are not correlated. Furthermore on a street, the headings are positively correlated (consider Manhattan one-way streets) and in this case a relay change due to a pass over will have more chances to arrive on a relay with good heading (one half instead of $\theta_c/\pi$). Again this would lead to less relay changes due to pass over.

The result still holds if we assume that the turn rate $\tau$ depends on $n$ and \mbox{$\tau=\tau_n=O(\log n)$}. In this case, the mobility model would fit even better for the realistic mobility of an urban area. Indeed the trajectories of cars should be fractal or self-similar, showing more frequent turns when cars are close to their physical destination (different than packet destination) or when leaving their parking lot. In this case, the overall turn rate tends to be in $O(\log n)$ with a coefficient depending on the Hurst parameter of the trajectory. This would lead to the same estimate of $O(\log n)$ relay change per packet. 

Figure~\ref{fig:fractal} illustrates a self-similar trajectory in an urban area. It shows a two-dimensional trajectory (upper half) and its traveled distance (lower half). The successive turns are indicated by $T_1,\ldots,T_7$. The trajectory after any turn $T_i$ looks like a reduced copy of the original trajectory. 
The CRB scheme may need some adaptation to cope with some unusual street configurations, 
{\it e.g.}, to replace the cartesian distance with the Manhattan distance in the street map.

\begin{figure} [!t]
\centering
\psfrag{a}{$T_1$}
\psfrag{b}{$T_2$}
\psfrag{c}{$T_3$}
\psfrag{d}{$T_4$}
\psfrag{e}{$T_5$}
\psfrag{f}{$T_6$}
\psfrag{g}{$T_7$}
\psfrag{h}{$T_4,...,T_7$}
\psfrag{i}{distance}
\includegraphics[scale=0.5]{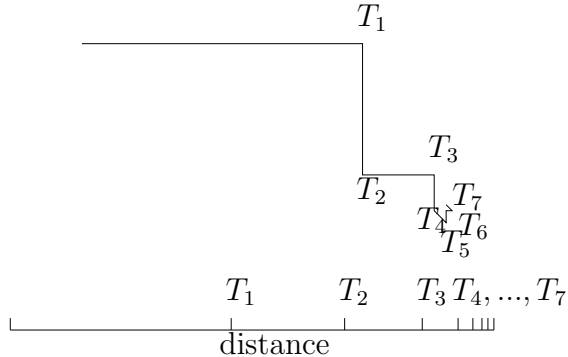}
\caption{Illustration of self-similar trajectories in urban areas.}
\label{fig:fractal}
\end{figure}

\section{Conclusions}
\label{sec:conclusions}

We have examined asymptotic capacity and delay in mobile networks with a georouting scheme, called CRB, for communication between source and destination nodes. Our results show that CRB allows to achieve the network capacity of $O(\frac{n}{\log n\log\log n})$ with packet delivery delay of $O(1)$ and transmissions per packet of $O(\log n)$. It is noticeable that this scheme does not need any sophisticated overhead for implementation. However, in this case, the mobile nodes must be aware of their position via a GPS system, for example. 

We have shown the asymptotic performance via analytical analysis under a unit disk graph model with random i.i.d. walks. The analytical results have been confirmed by simulations and in particular under ALOHA with SINR interference model. We have seen that the performance of CRB can be maintained even with non i.i.d. random walks, the latter being a worst case scenario. However, this latter result would require that the mobile nodes stay within same heading for $O(1/\log n)$ time. A next step would be to analyze the performance of this scheme on real traffic traces in urban areas. 

\bibliographystyle{hieeetr}
\bibliography{bibliography}

\end{document}